\documentclass[a4paper]{article}

\usepackage{graphicx}
\usepackage{natbib}
\usepackage{amsmath}
\usepackage{amsfonts}
\usepackage{url}

\newenvironment{proof}{\begin{trivlist} \item[]
{\bf Proof.}}{\nolinebreak
\hfill \rule{2mm}{2mm} \end{trivlist}}

\newcounter{ctr}

\newcounter{ctr1}

\newcounter{ctr2}

\newcounter{ctr3}

\newtheorem{definition}{Definition}[section]    
\newtheorem{theorem}[definition]{Theorem}
\newenvironment{theorem*}[1]{{\bf Theorem #1} \begin{itshape}}{\end{itshape}}

\newenvironment{corollary*}[1]{{\bf Corollary #1} \begin{itshape}}{\end{itshape}}

\newenvironment{proposition*}[1]{{\bf Proposition #1} \begin{itshape}}{\end{itshape}}

\newcommand{\RR}{\mathbb{R}}

\newcommand{\ud}{\, {\rm d} \kern-.015em }

\newcommand{\modulus}[1]{\left| \kern.05em #1 \kern.05em \right|}
\newcommand{\norm}[1]{\left\| \kern.05em #1 \kern.05em \right\|}
\newcommand{\inner}[1]{\left\langle \kern.05em #1 \kern.05em \right\rangle }

\newcommand{\pick}[2]{\renewcommand{\arraystretch}{0.6}
\left( \kern-.4em \begin{array}{c} #1 \\ #2 \end{array} \kern-.4em \right) }

\usepackage{colortbl}
\usepackage[ruled,vlined]{algorithm2e}
\usepackage{enumerate}
\usepackage{subfig}
\usepackage{stmaryrd}
\usepackage{natbib}

\usepackage{authblk}

\title{On Optimal Multiple Changepoint Algorithms for Large Data}
\author[1,$\dag$]{Robert Maidstone}
\author[2]{Toby Hocking}
\author[3]{Guillem Rigaill}
\author[4]{Paul Fearnhead} 
\affil[1]{STOR-i Centre for Doctoral Training, Lancaster University}
\affil[2]{McGill University and Genome Quebec Innovation Center}
\affil[3]{Unit\'{e} de Recherche en G\'{e}nomique V\'{e}g\'{e}tale (URGV), INRA-CNRS-Universit\'{e}
d'Evry}
\affil[4]{Department of Mathematics and Statistics, Lancaster University}
\affil[$\dag$]{Correspondence: r.maidstone@lancaster.ac.uk}

\begin{document}
\maketitle 
\begin{center}
 {\bf Abstract}
\end{center}

There is an increasing need for algorithms that can accurately detect changepoints in long time-series, or equivalent, data.
Many common approaches to detecting changepoints, for example based on penalised likelihood or minimum description length,
can be formulated in terms of minimising a cost over segmentations. Dynamic programming methods exist to solve this minimisation
problem exactly, but these tend to scale at least quadratically in the length of the time-series. Algorithms, such as Binary Segmentation,
exist that have a computational cost that is close to linear in the length of the time-series, but these are not guaranteed to find the 
optimal segmentation. Recently pruning
ideas have been suggested that can speed up the dynamic programming algorithms, whilst still being guaranteed to find true minimum of the cost
function. Here we extend these pruning methods, and introduce two new algorithms for segmenting data, FPOP and SNIP. Empirical results
show that FPOP is substantially faster than existing dynamic programming methods, and unlike the existing methods its computational
efficiency is robust to the number of changepoints in the data. We evaluate the method at detecting Copy Number Variations and observe
that FPOP has a computational cost that is competitive with that of Binary Segmentation.


 \section{Introduction}
\label{sec:introduction}

Often time series data experiences multiple abrupt changes in structure which need to be taken into account if the data is to be modelled effectively. These changes known as changepoints (or equivalently breakpoints) cause the data to be split into segments which can then be modelled seperately. Detecting changepoints, both accurately and efficiently, is required in a number of applications including financial data \citep{Fryzlewicz2012}, climate data \citep{Killick2012a, Reeves2007}, EEG data \citep{Lavielle2005} and the analysis of speech signals \citep{Davis2006}.

In Section \ref{sec:emp-eval-fpop} of this paper we look at detecting changes in DNA copy number in tumour microarray data. 
Regions in which this copy number is amplified or reduced from a baseline level can relate to tumorous cells and detection of these areas is crucial for classifying tumour progression and type. 
Changepoint methods are widely used in this area \citep{Zhang2007, Olshen2004, Picard2011} and moreover tumour microarray data has been used as a basis for benchmarking changepoint techniques both in terms of accuracy \citep{Hocking2013} and speed \citep{Hocking2014}.

Many approaches to estimating the number and position of changepoints \cite[e.g.][]{Braun/Braun/Muller:2000,Davis2006,Zhang2007} can be formulated in terms of defining a cost function for a segmentation. They then either minimise a penalised
version of this cost, which we call the penalised minimisation problem; or minimise the cost under a constraint on the number of changepoints, which we call the
constrained minimisation problem. If the cost function depends on the data through a sum of segment-specific costs then the minimisation can be done exactly using dynamic programming \cite[]{Auger1989,Jackson2005}. However
these dynamic programming methods have a cost that increases at least quadratically with the amount of data. This large computational cost is an increasingly important issue as every larger data sets are 
needed to be analysed.

Alternatively, much faster algorithms exist that provide approximate solutions to the minimisation problem.
The most widely used of these approximate techniques is Binary Segmentation \citep{Scott1974}. 
This takes a recursive approach, adding changepoints one at a time. With a new changepoint added in the position that would lead to the largest reduction in cost given the location of previous changepoints.
Due to its simplicity, Binary Segmentation is computationally efficient, being roughly linear in the amount of data, however it only provides an approximate solution and can lead to poor estimation
of the number and position of changepoints \cite[]{Killick2012a}.
Variations of Binary Segmentation, such as Circular Binary Segmentation \citep{Olshen2004} and Wild Binary Segmentation \citep{Fryzlewicz2012}, can offer more accurate solutions for slight decreases in the computational efficiency.

A better solution, if possible, is to look at ways of speeding up the dynamic programming algorithms. Recent work has shown this is possible via the pruning of the solution space. 
\citet{Killick2012a} present a technique for doing this which we shall refer to as {\em inequality based pruning}. This forms the basis of their method PELT which can be used to solve the penalised minimisation problem.  
\citet{Rigaill2010} develop a different pruning technique, {\em functional pruning}, and this is used in their pDPA method which can be used to solve the constained minisation problem.
Both PELT and pDPA are optimal algorithms, in the sense that they
find the true optimum of the minimisation problem they are trying to solve. 
However the pruning approaches they take are very different, and work well in different scenarios. PELT is most efficient
in applications where the number of changepoints is large, and pDPA when there are few changepoints. 

The focus of this paper is on these pruning techniques, with the aim of trying to combine ideas from PELT and pDPA. This leads to two new algorithms, FPOP and SNIP. The former uses functional pruning to solve the 
penalised minisation problem, and the latter uses inequality based pruning to solve the constrainted minisation problem. We further show that FPOP always prunes more than PELT. Empirical results suggest that FPOP is efficient for
large data sets regardless of the number of changepoints, and we observe that FPOP has a computational cost that is even competitive with Binary Segmentation.

The structure of the paper is as follows. We introduce the constrained and penalised optimisation problems for segmenting data in the next section. We then review the existing dynamic progamming methods and pruning approaches 
for solving the penalised optimisation problem in Section \ref{sec:pen} and for solving the constrained optimisation problem in Section \ref{sec:const}. The new algorithms, FPOP and SNIP, are developed in
Section \ref{sec:incr-effic}, and compared empirically and theoretically with existing pruning methods in Section \ref{sec:simil-diff-betw}. We then evaluate FPOP empirically on both simulated and CNV data in Section 
\ref{sec:emp-eval-fpop}. The paper ends with a discussion.



\section{Model Definition}
\label{sec:defining-change-mean}
Assume we have data ordered by time, though the same ideas extend trivially to data ordered by any other attribute such as position along a chromosome.
Denote the data by $\mathbf{y}=(y_1,\hdots,y_n)$. We will use that notation that, for $s\geq t$, the set of observations from time $t$ to time $s$ is $\mathbf{y}_{t:s}=(y_{t},...,y_s)$. 
If we assume that there are $k$ changepoints in the data, this will correspond to the data being split into $k+1$ distinct segments. 
We let the location of the $j$th changepoint be $\tau_j$ for $j=1,\hdots,k$, and set $\tau_0=0$ and $\tau_{k+1}=n$. The $j$th segment will consist of data points $y_{\tau_{j-1}+1},\ldots,y_{\tau_j}$. We let 
$\mathbf{\tau}=(\tau_0,\ldots,\tau_{k+1})$ be the set of changepoints.

The statistical problem we are considering is how to infer both the number of changepoints and their locations. The specific details of any approach will depend on the type of change, such as change in mean, variance or distribution, that we wish to detect. However a general framework that encompasses many changepoint detection methods is to introduce a cost function for each segment. The cost of a segmentation can then be defined in terms of the sum of the costs across the segments, and we can infer segmentations through minimising the segmentation cost.

Throughout we will let $\mathcal{C}(\mathbf{y}_{t+1:s})$, for $s\leq t$, denote the cost for a segment consisting of data points $y_{t+1},\ldots,y_s$. The cost of a segmentation, $\tau_1,\ldots,\tau_k$ is then 
\begin{align}
 \sum^k_{j=0}\mathcal{C}(\mathbf{y}_{\tau_j+1:\tau_{j+1}}) .\label{eq:1a}
\end{align}

The form of this cost function will depend on
the type of change we are wanting to detect. One generic approach to defining these segments is to introduce a model for the data within a segment, and then to let the cost be minus the maximum log-likelihood for the data in that segment. If our model is that the data is independent and identically distributed with segment-specific parameter $\mu$ then 
\begin{align} \label{eq:nll}
  \mathcal{C}(\mathbf{y}_{t+1:s})=\min_{\mu}\sum_{i=t+1}^{s}-\log (p(y_i|\mu)).
\end{align}
In this formulation we are detecting changes in the value of the parameter, $\mu$, across segments.

For example if $\mu$ is the mean in Normally distributed data, with known variance $\sigma^2$, then the cost for a segment would simply be
\begin{align} \label{eq:Cost}
  \mathcal{C}(\mathbf{y}_{t+1:s})=
  \frac{1}{2\sigma^2}\left[\min_{\mu}\sum_{i=t+1}^s\left(y_i-\frac{1}{s-t}\sum_{j=t+1}^s y_j\right)^2\right],
\end{align}
which is just a quadratic error loss. We have removed a term that does not depend on the data and is linear in segment length, as this term does not affect the optimal segmentation.
Note we will get the same optimal segmentation for any choice of $\sigma>0$ using this cost function.

\subsection{Finding the Optimal Segmentation}
\label{sec:find-optim-segm}

If we knew the number of changepoints in the data, $k$, then we could infer their location through minimising (\ref{eq:1a}) over all segmentations with $k$ changepoints.  Normally however $k$ is unknown, and also has to be estimated. A common approach is to define
\begin{align}
 C_{k,n}=\min_{\boldsymbol{\tau}}\left[\sum^k_{j=0}\mathcal{C}(\mathbf{y}_{\tau_j+1:\tau_{j+1}})\right],\label{eq:1}
\end{align}
the minimum cost of a segmenting data $y_{1:n}$ with $k$ changepoints. As $k$ increases we have more flexibility in our model for the data, so often $C_{k,n}$ will be monotonically decreasing in $k$ and estimating the number of changepoints by minimising $C_{k,n}$ is not possible. One solution is to solve (\ref{eq:1}) for a fixed value of $k$  which is either assumed to be known or chosen separately. We call this problem the {\em constrained case}.

If $k$ is not known, then a common approach is to calculate $C_{k,n}$ and the corresponding optimal segmentations for a range of values, $k=0,1,\ldots,K$, where $K$ is some chosen maximum number. We can then
estimate the number of changepoints by minimising $C_{k,n}+f(k,n)$ over $k$ for some suitable penalty function $f(k,n)$. 
The most common choices of $f(k,n)$, for example SIC \citep{Schwarz1978} and AIC \citep{Akaike1974} are linear in $k$

If the penalty function is linear in $k$, with $f(k,n)=\beta k$ for some $\beta>0$ (which may depend on $n$), then we can directly find the optimal number of changepoints and segmentation by noting that
\begin{eqnarray}
\min_k \left[C_{k,n}+\beta k\right] &=& \min_{k,\boldsymbol{\tau}}\left[\sum^k_{j=0}\mathcal{C}(\mathbf{y}_{\tau_j+1:\tau_{j+1}})\right]+\beta k \nonumber \\
 &=&  \min_{k,\boldsymbol{\tau}}\left[\sum^k_{j=0}\mathcal{C}(\mathbf{y}_{\tau_j+1:\tau_{j+1}})+\beta \right] -\beta.\label{eq:2}
\end{eqnarray}
We call the minimisation problem in (\ref{eq:2}) the {\em penalised case}.

In both the constrained and penalised cases we need to solve a minimisation problem to find the optimal segmentation under our criteria. There are efficient dynamic programming algorithms for solving each of these minimisation problems. 
For the constrained case this is achieved using the Segment Neighbourhood Search algorithm (see Section~\ref{sec:seg-neigh-search}), whilst for the penalised case 
this can be achieved using the Optimal Partitioning algorithm (see Section~\ref{sec:optim-part}).

Solving the constrained case offers a way to get optimal segmentations for $k=0,1,\hdots,K$ changepoints, and thus gives insight into how the segmentation varies with the number of segments.
However, a big advantage of the penalised case is that it incorporates model selection into the problem itself, and therefore is often computationally more efficient when dealing with an unknown value of $k$. 

\subsection{Conditions for Pruning}

The focus of this paper is on methods for speeding up these dynamic programming algorithms using pruning methods. The pruning methods can be applied under one of two conditions on the segment costs:

\begin{description}
\item[C1] The cost function satisfies
\begin{align*}
  \mathcal{C}(\mathbf{y}_{t+1:s})=\min_{\mu}\sum_{i=t+1}^{s}\gamma(y_i,\mu),
\end{align*}
for some function $\gamma(\cdot,\cdot)$, with parameter $\mu$.
\item[C2] There exists a constant $\kappa$ such that for all $t<s<T$,
  \begin{align*}
      \mathcal{C}(\mathbf{y}_{t+1:s})+\mathcal{C}(\mathbf{y}_{s+1:T})+\kappa \leq\mathcal{C}(\mathbf{y}_{t+1:T}).
  \end{align*}
\end{description}

Condition C1 will be used by functional pruning (which is discussed in Sections \ref{sec:segm-neighb-search} and \ref{sec:opr}).
Condition C2  will be used by the inequality based pruning (Section~\ref{sec:optim-part-pelt} and~\ref{sec:snp}). Note that C1 is a stronger
condition than C2. If C1 holds then C2 also holds with $\kappa=0$.

For many practical cost functions these conditions hold; for example it is easily seen that for the negative log-likelihood (\ref{eq:nll}) C1 holds with $\gamma(y_t,\mu)=-\log(p(y_t|\mu))$ and C2 holds with $\kappa=0$.

\section{Solving the Penalised Optimisation Problem}\label{sec:pen}

We first consider solving the penalised optimisation problem \eqref{eq:2} using a dynamic programming approach. The initial algorithm, Optimal Partitioning \citep{Jackson2005}, will be discussed first before mentioning how pruning can be used to reduce the computational cost.

\subsection{Optimal Partitioning}
\label{sec:optim-part}
Consider segmenting the data $\mathbf{y}_{1:t}$. Denote $F(t)$ to be the minimum value of the penalised cost (\ref{eq:2}) for segmenting such data, with $F(0)=-\beta$. 
The idea of Optimal Partitioning is to split the minimisation over segmentations into the minimisation over the position of the last changepoint, and then the minimisation over the earlier changepoints. 
We can then use the fact that the minimisation over the earlier changepoints will give us the value $F(\tau^*)$ for some $\tau^*<t$
\begin{align*}
  F(t)&=\min_{\boldsymbol{\tau},k}\sum_{j=0}^k\left[\mathcal{C}(\mathbf{y}_{\tau_j+1:\tau_{j+1}})+\beta\right]-\beta,\\
&=\min_{\boldsymbol{\tau},k}\left\{\sum_{j=0}^{k-1}\left[\mathcal{C}(\mathbf{y}_{\tau_j+1:\tau_{j+1}})+\beta\right]+\mathcal{C}(\mathbf{y}_{\tau_k+1:t})+\beta\right\}-\beta,\\
&=\min_{\tau^*}\left\{\min_{\boldsymbol{\tau},k'}\sum_{j=0}^{k'}\left[\mathcal{C}(\mathbf{y}_{\tau_j+1:\tau_{j+1}})+\beta\right]-\beta+\mathcal{C}(\mathbf{y}_{\tau^*+1:t})+\beta\right\},\\
&=\min_{\tau^*}\left\{F(\tau^*)+\mathcal{C}(\mathbf{y}_{\tau^*+1:t})+\beta\right\}.
\end{align*}

Hence we obtain a simple recursion for the $F(t)$ values
\begin{align}
 F(t)=\min_{0\leq\tau<t}\left[F(\tau) + \mathcal{C}(\mathbf{y}_{\tau+1:t}) + \beta\right]. \label{eq:4}
\end{align}

The segmentations themselves can recovered by first taking the arguments which minimise (\ref{eq:4})
\begin{align}
  \label{eq:6}
  \tau^*_t=\operatorname*{arg\,min}_{0\leq\tau<t}\left[F(\tau) + \mathcal{C}(\mathbf{y}_{\tau+1:t}) + \beta\right],
\end{align}
which give the optimal location of the last changepoint in the segmentation of $y_{1:t}$. 

 If we denote the vector of ordered changepoints in the optimal segmentation of $y_{1:t}$ by $cp(t)$, with $cp(0)=\emptyset$, then the optimal changepoints up to a time $t$ can be calculated recursively
\begin{align*}
  cp(t)=(cp(\tau^*_t),\tau^*_t).
\end{align*}
As equation~\eqref{eq:4} is calculated for time steps $t=1,2,\hdots,n$ and each time step involves a minimisation over $\tau=0,1,\hdots,t-1$ 
the computation actually takes $\mathcal{O}(n^2)$ time.

\subsection{PELT} \label{sec:optim-part-pelt}

One way to increase the efficiency of Optimal Partitioning is discussed in \citet{Killick2012a} where they introduce the PELT (Pruned Exact Linear Time) algorithm. PELT works by limiting the set of potential 
previous changepoints (i.e. the set over which $\tau$ is chosen from in the minimisation in equation~\ref{eq:4}). They show that if Condition C2 holds for some $\kappa$, and if
\begin{align} 
  F(t)+\mathcal{C}(\mathbf{y}_{(t+1:s)})+\kappa > F(s),\label{eq:hold}
\end{align}
then at any future time $T>s$, $t$ can never be the optimal location of the most recent changepoint prior to $T$.



This means that at every time step $s$ the left hand side of equation~(\ref{eq:hold}) can be calculated for all potential values of the last changepoint. 
If the inequality holds for any individual $t$ then that $t$ can be discounted as a potential last changepoint for all future times.
Thus the update rules (\ref{eq:4}) and (\ref{eq:6}) can be restricted to a reduced set of potential last changepoints, $\tau$, to consider. 
This set, which we shall denote as $R_t$, can be updated simply by
\begin{align}
  R_{t+1}=\{\tau\in \{R_{t}\cup\{t\}\}:F(\tau)+\mathcal{C}(\mathbf{y}_{(\tau+1):t})+\kappa \leq F(t)\}.
\end{align}
This pruning technique, which we shall refer to as {\em inequality based pruning}, forms the basis of the PELT method.


As at each time step in the PELT algorithm the minimisation is being run over fewer values it would be expected that this method would be more efficient than the basic Optimal Partitioning algorithm. 
In \citet{Killick2012a} it is shown to be at least as efficient as Optimal Partitioning, with PELT's computational cost being bounded above by $\mathcal{O}(n^2)$. Under certain conditions the 
expected computational cost can be shown to be bounded by $Ln$ for some constant $L<\infty$. These conditions are given fully in \citet{Killick2012a}, the most important of which is that the expected number of changepoints in the data increases linearly with the length of the data, $n$.

\section{Solving the Constrained Optimisation Problem} \label{sec:const}

We now consider solving the constrained optimisation problem (\ref{eq:1}) using dynamic programming. These methods assume a maximum number of changepoints that are to be considered, $K$, and then solve the constrained optimisation problem for all values of $k=1,2,\ldots,K$. We first describe the initial algorithm,  Segment Neighbourhood Search \citep{Auger1989}, and then an approach that uses pruning. 

\subsection{Segment Neighbourhood Search}
\label{sec:seg-neigh-search}
Take the constrained case \eqref{eq:1} which segments the data up to $t$, for $t\geq k+1$, into $k+1$ segments (using $k$ changepoints), and denote the minimum value of the cost by $C_{k,t}$. 
The idea of Segment Neighbourhood Search is to derive a relationship between $C_{t,k}$ and $C_{s,k-1}$ for $s<t$:
\begin{align*}
  C_{k,t}&=\min_{\boldsymbol{\tau}}\sum_{j=0}^k\mathcal{C}(\mathbf{y}_{\tau_j+1:\tau_{j+1}}),\\
&= \min_{{\tau_k}}\left[\min_{\boldsymbol{\tau}_{1:k-1}}\sum_{j=0}^{k-1}\mathcal{C}(\mathbf{y}_{\tau_j+1:\tau_{j+1}})+\mathcal{C}(\mathbf{y}_{\tau_k+1:\tau_{k+1}})\right],\\
&=\min_{{\tau_k}}\left[C_{k-1,\tau_k}+\mathcal{C}(\mathbf{y}_{\tau_k+1:\tau_{k+1}})\right].
\end{align*}
Thus the following recursion is obtained:
\begin{align}
  C_{k,t}=\min_{\tau\in\{k,\hdots,t-1\}}\left[C_{k-1,\tau}+\mathcal{C}(\mathbf{y}_{\tau+1:t})\right].\label{eq:SNS}
\end{align}
If this is run for all values of $t$ up to $n$ and for $k=2,\hdots,K$, then the optimal segmentations with $1,\hdots,K$ segments can be acquired.

To extract the optimal segmentation we first let $\tau^*_l(t)$ denote the optimal position of the last changepoint if we segment data $\mathbf{y}_{1:t}$ using $l$ changepoints. This can be calculated as
\[
 \tau^*_l(t)=\operatorname*{arg\,min}_{\tau\in\{l,\hdots,t-1\}}\left[C_{l-1,\tau}+\mathcal{C}(\mathbf{y}_{\tau+1:t})\right].
\]
Then if we let $(\tau_1^k,\ldots,\tau_k^k)$ be the set of changepoints in the optimal segmentation of $\mathbf{y}_{1:n}$ into $k+1$ segments, we have $\tau_k^k=\tau^*_k(n)$. Furthermore we can calculate the other
changepoint positions recursively for $l=k-1,\ldots,1$ using
\[
 \tau_l^k(n)=\tau^*_l(\tau_{l+1}^k).
\]

For a fixed value of $k$ equation~(\ref{eq:SNS}) is computed for $t\in 1,\hdots,n$. Then for each $t$ the minimisation is done for $\tau=1,\hdots,t-1$. 
This means that $\mathcal{O}(n^2)$ calculations are needed. 
However, to also identify the optimal number of changepoints this then needs to be done for $k\in 1,\hdots,K$ so the total computational cost in time can be seen to be $\mathcal{O}(Kn^2)$.

\subsection{Pruned Segment Neighbourhood Search} \label{sec:segm-neighb-search}

\citet{Rigaill2010} has developed techniques to increase the efficiency of Segment Neighbourhood Search using functional pruning. 
These form the basis of a method called pruned Dynamic Programming Algorithm (pDPA). A more generic implementation of this method is presented in \citet{Cleynen2012}.
Here we describe how this algorithm can be used to calculate the $C_{k,t}$ values. Once these are calculated, the optimal segmentation can be extracted as in Segment Neighbourhood Search.

Assuming condition C1, the segment cost function can be split into the component parts $\gamma(y_i,\mu)$, which depend on the parameter $\mu$. 
We can then define new cost functions, $Cost^{\tau}_{k,t}(\mu)$, as the minimal cost of segmenting data $y_{1:t}$ into $k$ segments, with a most
recent changepoint at $\tau$, and where the segment after $\tau$ is conditioned to have parameter $\mu$. Thus for $\tau\leq t-1$,
\begin{align}
  Cost^\tau_{k,t}(\mu)=C_{k-1,\tau}+\sum_{i=\tau+1}^t\gamma(y_i,\mu),\label{eq:costfunk}
\end{align}
and $Cost^t_{k,t}(\mu)=C_{k-1,t}$.

These functions, which are stored for each candidate changepoint, can then be updated at each new time step as for $\tau\leq t-1$
\begin{align}
  Cost^\tau_{k,t}(\mu)=Cost^\tau_{k,t-1}(\mu)+\gamma(y_{t},\mu).\label{eq:update}
\end{align}

By taking the minimum of $Cost^\tau_{k,t}(\mu)$ over $\mu$, the individual terms of the right hand side of equation~\eqref{eq:SNS} can be recovered. 
Therefore, by further minimising over $\tau$, the minimum cost $C_{k,t}$ can be returned
\begin{align*}
  \min_\tau \min_\mu Cost^\tau_{k,t}(\mu)&= \min_\tau \min_\mu \left[C_{k-1,\tau}+\sum_{i=\tau+1}^t\gamma(y_i,\mu)\right],\\
&=\min_\tau \left[C_{k-1,\tau}+\min_\mu\sum_{i=\tau+1}^t\gamma(y_i,\mu)\right],\\
&=\min_\tau \left[C_{k-1,\tau}+\mathcal{C}(\mathbf{y}_{\tau+1:t})\right],\\
&= C_{k,t}.
\end{align*}

By interchanging the order of minimisation the values of the potential last changepoint, $\tau$, can be pruned whilst allowing for changes in $\mu$. First we define the function $Cost^*_{k,t}(\mu)$ as follows
\begin{align*}
  Cost^*_{k,t}(\mu)= \min_{\tau}Cost^\tau_{k,t}(\mu).
\end{align*}
We can now get a recursion for $Cost^*_{k,t}(\mu)$ by splitting the minimisation over the most recent changepoint $\tau$ into the two cases $\tau\leq t-1$ and $\tau=t$:
\begin{eqnarray*}
  Cost^*_{k,t}(\mu)&=& \min\left\{ \min_{\tau\leq t-1} Cost^\tau_{k,t}(\mu)\mbox{ },\mbox{ }Cost^t_{k,t}(\mu) \right\}\\
&=&\min\left\{ \min_{\tau\leq t-1} Cost^\tau_{k,t-1}(\mu)+\gamma(y_t,\mu)\mbox{ },\mbox{ }C_{k-1,t} \right\},
\end{eqnarray*}
which gives
\[
Cost^*_{k,t}(\mu)=\min\left\{  Cost^*_{k,t-1}(\mu)+\gamma(y_t,\mu)\mbox{ },\mbox{ }C_{k-1,t} \right\}.
\]

The idea of pDPA is the use this recursion for $Cost^*_{k,t}(\mu)$. We can then use the fact that $C_{k,t}=\min_\mu Cost^*_{k,t}(\mu)$ to calculate the $C_{k,t}$ values. 
 In order to do this we need to be able to represent this function of $\mu$ in an efficient way.
 This can be done if $\mu$ is a scalar, because for any value of $\mu$, $Cost^*_{k,t}(\mu)$  is equal to the value of $Cost^\tau_{k,t}(\mu)$ for some value of $\tau$. 
 Thus we can partition the possible values of $\mu$ into intervals, with each interval corresponding to a value for $\tau$ for which $Cost^*_{k,t}(\mu)=Cost^\tau_{k,t}(\mu)$. 
 
 To make the idea concrete, an example of  $Cost^*_{k,t}(\mu)$ is given in  Figure~\ref{fig:pDPAfuncallin} for a change in mean using a least square cost criteria. 
 Each $Cost^\tau_{k,t}(\mu)$ is a quadratic function in this example. 
In this example there are 6 intervals of $\mu$ corresponding to 5 different values of $\tau$ for which $Cost^*_{k,t}(\mu)=Cost^\tau_{k,t}(\mu)$.  
The pDPA algorithm needs to just store the 5 different $Cost^\tau_{k,t}(\mu)$ functions, and the corresponding sets.

 \begin{figure}[t]
  \centering
  \includegraphics[width=9cm, trim=1cm 1cm 1cm 1cm]{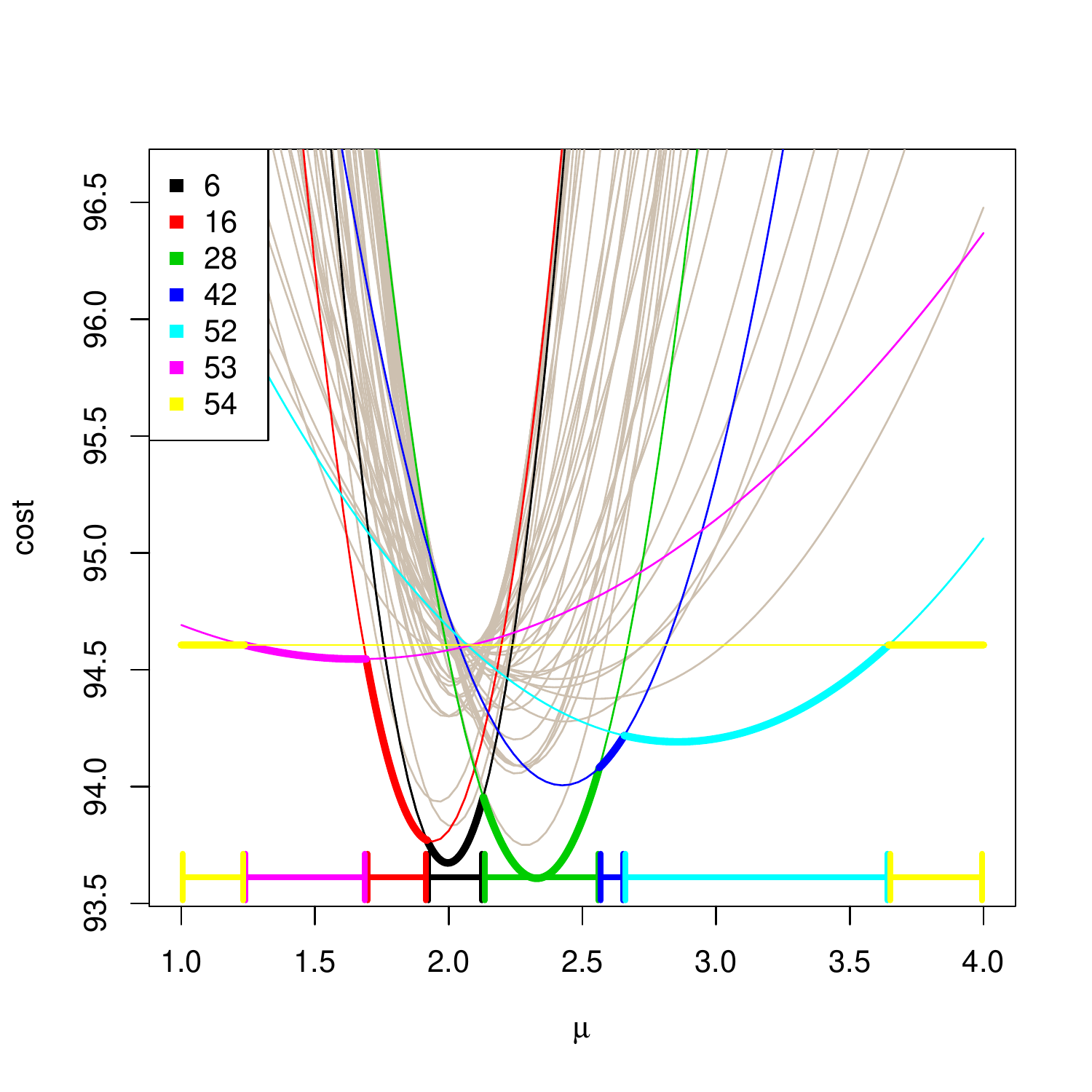}
  \caption{Cost functions, $Cost_{k,\tau}(\mu,t)$ for $\tau=0,\hdots,46$ and $t=46$ and the corresponding $C^*_{k}(\mu,t)$ (in bold) for a change in mean using a least square cost criteria.
  Coloured lines correspond to $Cost_{k,\tau}(\mu,t)$ that contribute to $C^*_{k}(\mu,t)$, with the coloured horizontal lines showing the intervals of $\mu$ for which each value of $\tau$ is such that
  $Cost_{k,\tau}(\mu,t)=C^*_{k}(\mu,t)$.  Greyed out lines correspond to candidates which have previously been pruned, and do not contribute to $C^*_{k}(\mu,t)$. 
 \label{fig:pDPAfuncallin}}
\end{figure}

Formally speaking we define the set of intervals for which $Cost^*_{k,t}(\mu)=Cost^\tau_{k,t}(\mu)$ as $Set_{k,t}^\tau$. The recursion for $Cost^*_{k,t}(\mu)$ can be used to induce a recursion
for these sets. First define:
\begin{align}
 I^\tau_{k,t}=\{\mu: Cost^\tau_{k,t}(\mu)\leq C_{k-1,t}\}. \label{eq:pDPAprune}
\end{align}
Then, for $\tau\leq t-1$ we have
\begin{eqnarray*}
 Set_{k,t}^{\tau}&=&\left\{\mu:Cost^\tau_{k,t}(\mu)= Cost^*_{k,t}(\mu)\right\}\\
&=&\left\{\mu:Cost^\tau_{k,t-1}(\mu)+\gamma(y_t,\mu)=\min\left\{  Cost^*_{k,t-1}(\mu)+\gamma(y_t,\mu),C_{k-1,t} \right\} \right\}.
\end{eqnarray*} 
Remembering that $Cost^\tau_{k,t-1}(\mu)+\gamma(y_t,\mu)\geq Cost^*_{k,t-1}(\mu)+\gamma(y_t,\mu)$, we have that
for $\mu$ to be in $Set_{k,t}^{\tau}$ we need that $Cost^\tau_{k,t-1}(\mu)=Cost^*_{k,t-1}(\mu)$, and that $Cost^\tau_{k,t-1}(\mu)+\gamma(y_t,\mu)\leq C_{k-1,t}$. 
The former condition corresponds to $\mu$ being  in $Set_{k,t-1}^{\tau}$ and the second that $\mu$ is in $I^\tau_{k,t}$. So for $\tau\leq t-1$
\[ 
Set^{\tau}_{k,t}=Set^{\tau}_{k,t-1} \cap I^\tau_{k,t}.
\]
If this $Set^{\tau}_{k,t}=\emptyset$ then the value $\tau$ can be pruned, as $Set^{\tau}_{k,T}=\emptyset$ for all $T>t$.

If we denote the range of values $\mu$ can take to be $D$, then we further have that
\[
Set^{t}_{k,t}=D \backslash \left[\displaystyle\bigcup_{\tau}I^\tau_{k,t}\right],
\]
where $t$ can be pruned straight away if $Set^t_{k,t}=\emptyset$.



\begin{figure}[t]
  \centering
  \subfloat[End Time $t$]{
  \includegraphics[page=1,width=5.5cm,trim=1cm .5cm .5cm 1.5cm]{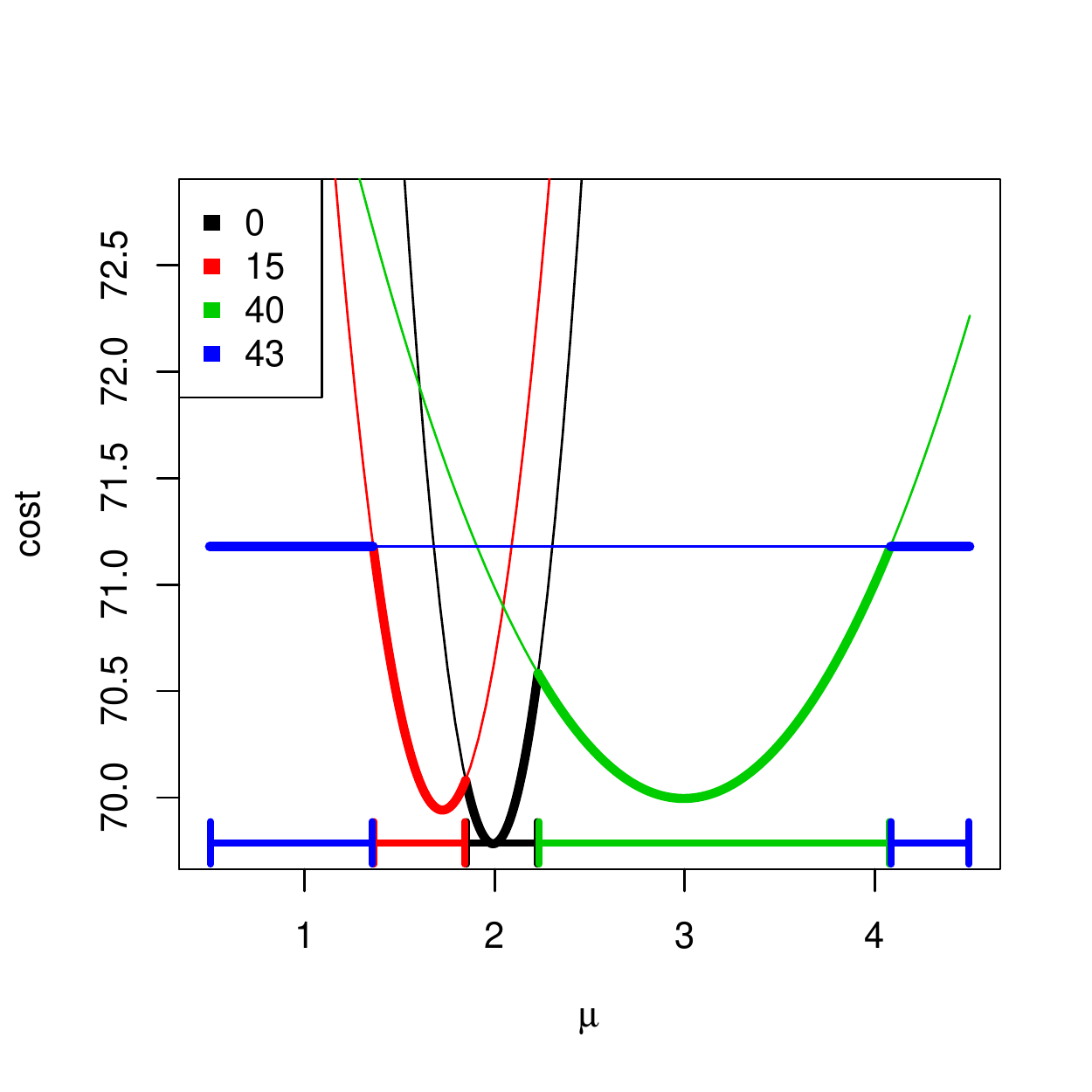}
} 
 \subfloat[Middle Time $t+1$]{
  \includegraphics[page=2,width=5.5cm,trim=.75cm .5cm .75cm 1.5cm]{pDPA/PDPAcandidates.pdf}
}  
 \subfloat[End Time $t+1$]{
  \includegraphics[page=3,width=5.5cm,trim=.5cm .5cm 1cm 1.5cm]{pDPA/PDPAcandidates.pdf}
}  
  \caption{Example of pDPA algorithm over two time-steps. On each plot we show individual $Cost^\tau_{k,t}(\mu)$ functions that are stored, 
  together with the intervals (along the bottom) for which each candidate last changepoint is optimal. In bold is the value of $Cost^*_{k,t}(\mu)$. 
  For this example $t=43$ and we are detecting a change in mean (see Section \ref{sec:defining-change-mean}).  (a) 4 candidates are optimal for some interval of $\mu$, however at $t=44$ (b), 
  when the candidate functions are updated and the new candidate is added, then the candidate $\tau=43$ is no longer optimal for any $\mu$ and hence can be pruned (c).
  \label{fig:pDPAfuncs} }
\end{figure}

An example of the pDPA recursion is given in Figure~\ref{fig:pDPAfuncs} for a change in mean using a least square cost criteria. 
The left-hand plot shows $Cost^*_{k,t}(\mu)$. 
In this example there are 5 intervals of $\mu$ corresponding to 4 different values of $\tau$ for which $Cost^*_{k,t}(\mu)=Cost^\tau_{k,t}(\mu)$.  
When we analyse the next data point, we update each of these four $Cost^\tau_{k,t}(\mu)$ functions, using $Cost^\tau_{k,t+1}(\mu)=Cost^\tau_{k,t}(\mu)+\gamma(y_{t+1},\mu)$, and introduce a new curve corresponding to a change-point at time $t+1$,
$Cost^{t+1}_{k,t+1}(\mu)=C_{k-1,t+1}$ (see middle plot). We can then prune the functions which are no longer optimal for any $\mu$ values, and in this case we remove one such function (see right-hand plot).

pDPA can be shown to be bounded in time by $\mathcal{O}(Kn^2)$. 
\citet{Rigaill2010} further analyses the time complexity of pDPA and shows it empirically to be $\mathcal{O}(Kn\log n)$, further indications towards this will be presented in Section~\ref{sec:emp-eval-fpop}. 
However pDPA has a computational overhead relative to Segment Neighbourhood Search, as it requires calculating and storing the $Cost^\tau_{k,t}(\mu)$ functions and the corresponding sets $Set_{k,t}^\tau$. 
Currently implementations of pDPA have only been possible for models with scalar segment parameters $\mu$, due to the difficulty of calculating the sets in higher dimensions. 
Being able to efficiently store and update the $Cost^\tau_{k,t}(\mu)$ have also restricted application primarily to models where $\gamma(y,\mu)$ corresponds to the log-likelihood of an exponential family. 
However this still includes a wide-range of changepoint applications, including that of detecting CNVs that we consider in Section \ref{sec:emp-eval-fpop}.
The cost of updating the sets depends heavily on whether the updates \eqref{eq:pDPAprune} can be calculated analytically, or whether they require the use of numerical methods.

\section{New Changepoint Algorithms} 
\label{sec:incr-effic}
Two natural ways of extending the two methods introduced above will be examined in this section. 
These are, respectively, to apply functional pruning (Section~\ref{sec:segm-neighb-search}) to Optimal Partitioning, and to apply inequality based pruning (Section~\ref{sec:optim-part-pelt}) to Segment Neighbourhood Search. 
These lead to two new algorithms, which we call Functional Pruning Optimal Partitioning (FPOP) and Segment Neighbourhood with Inequality Pruning (SNIP). 

\subsection{Functional Pruning Optimal Partitioning}
\label{sec:opr}
Functional Pruning Optimal Partitioning (FPOP) provides a version of Optimal Partitioning \citep{Jackson2005} which utilises functional pruning to increase the efficiency. As will be discussed in Section~\ref{sec:simil-diff-betw} and shown in Section~\ref{sec:emp-eval-fpop} FPOP provides an alternative to PELT which is more efficient in certain scenarios. The approach used by FPOP is similar to the approach for pDPA in Section~\ref{sec:segm-neighb-search}, however the theory is slightly simpler here as there is no longer the need to condition on the number of changepoints.

We assume condition C1 holds, that the cost function, $\mathcal{C}(\mathbf{y}_{\tau+1:t})$, can be split into component parts $\gamma(y_i,\mu)$ which depend on the parameter $\mu$. Cost functions $Cost_t^\tau$ can then be defined as the minimal cost of the data up to time $t$, conditional on the last changepoint being at $\tau$ and the last segment having parameter $\mu$. Thus for $\tau\leq t-1$
\begin{align}
  \label{eq:3}
  Cost_t^\tau(\mu)=F(\tau)+\beta+\sum_{i=\tau+1}^t\gamma(y_i,\mu),
\end{align}
and $Cost_t^t(\mu)=F(t)+\beta$.

These functions which only need to be stored for each candidate changepoint can then be recursively updated at each time step, $\tau\leq t-1$
\begin{align}
  \label{eq:8}
  Cost^\tau_{t}(\mu)=Cost^\tau_{t-1}(\mu)+\gamma(y_{t},\mu).
\end{align}
Given the cost functions $Cost^\tau_t(\mu)$ the minimal cost $F(t)$ can be returned by minimising over both $\tau$ and $\mu$:
\begin{align*}
  \label{eq:9}
  \min_\tau\min_\mu Cost_t^\tau(\mu) &= \min_\tau\min_\mu\left[F(\tau)+\beta+\sum_{i=\tau+1}^t\gamma(y_i,\mu)\right],\\
&=\min_\tau\left[F(\tau)+\beta+\min_\mu\sum_{i=\tau+1}^t\gamma(y_i,\mu)\right],\\
&=\min_\tau\left[F(\tau)+\beta+\mathcal{C}(\mathbf{y}_{\tau+1:t})\right],\\
&=F(t).
\end{align*}

As before, by interchanging the order of minimisation, the values of the potential last changepoint, $\tau$, can be pruned whilst allowing for a varying $\mu$. Firstly we will define the function $Cost_{t}^*(\mu)$, the minimal
cost of segmenting data $y_{1:t}$ conditional on the last segment having parameter $\mu$:
\[
 Cost^*_t(\mu)=\min_{\tau} Cost^\tau_t(\mu). 
\]
We will update these functions recursively over time, and use $F(t)=\min_\mu Cost^*_t(\mu)$ to then obtain the solution of the penalised minimisation problem. The recursions for $Cost^*_t(\mu)$ are obtained by splitting the minimisation over $\tau$ into $\tau\leq t-1$ and $\tau=t$
\begin{align*}
 Cost^*_t(\mu)&=\min\left\{ \min_{\tau \leq t-1} Cost^\tau_t(\mu)\mbox{ },\mbox{ }Cost^t_t(\mu)\right\},\\
&=\min\left\{ \min_{\tau \leq t-1} Cost^\tau_{t-1}(\mu)+\gamma(y_t,\mu)\mbox{ },\mbox{ }Cost^t_t(\mu)\right\},
\end{align*}
which then gives:
\[
 Cost^*_t(\mu)=\min\{ Cost^*_{t-1}(\mu)+\gamma(y_t,\mu)\mbox{ },\mbox{ }F(t)+\beta\}.
\]

To implement this recursion we need to be able to efficiently store and update $Cost^*_t(\mu)$. As before we do this by partitioning the space of possible $\mu$ values, $D$, into sets where each set corresponds to
a value $\tau$ for which $Cost^*_t(\mu)=Cost^\tau_t(\mu)$. We then need to be able to update these sets, and store $Cost^\tau_t(\mu)$ just for each $\tau$ for which the corresponding set is non-empty.

This can be achieved by first defining
\begin{align}
  \label{eq:5}
  I^\tau_{t}=\{\mu: Cost^\tau_{t}(\mu)\leq F(t) +\beta\}.
\end{align}
Then, for $\tau\leq t-1$, we define
\begin{align*}
  Set_t^\tau &= \{\mu: Cost_t^\tau(\mu)=Cost_t^*(\mu)\}\\
&=\{\mu: Cost_{t-1}^\tau(\mu)+\gamma(y_t,\mu)=\min{\{Cost_{t-1}^*(\mu)+\gamma(y_t,\mu)\mbox{ },\mbox{ }F(t)+\beta\}}\}
\end{align*}

Remembering that $Cost^\tau_{t-1}(\mu)+\gamma(y_t,\mu)\geq Cost^*_{t-1}(\mu)+\gamma(y_t,\mu)$, we have that
for $\mu$ to be in $Set_{t}^{\tau}$ we need that $Cost^\tau_{t-1}(\mu)=Cost^*_{t-1}(\mu)$, and that $Cost^\tau_{t-1}(\mu)+\gamma(y_t,\mu)\leq F(t)+\beta$. 
The former condition corresponds to $\mu$ being  in $Set_{t-1}^{\tau}$ and the second that $\mu$ is in $I^\tau_{t}$, so for $\tau\leq t-1$
\[
Set^{\tau}_{t}=Set^{\tau}_{t-1} \cap I^\tau_{t}.
\]
If this $Set^{\tau}_{t}=\emptyset$ then the value $\tau$ can be pruned, as then $Set^{\tau}_{T}=\emptyset$ for all $T>t$.

If we denote the range of values $\mu$ can take to be $D$, then we further have that
\[
Set^{t}_{t}=D \backslash \left[\displaystyle\bigcup_{\tau}I^\tau_{t}\right],
\]
where $t$ can be pruned straight away if $Set_t^t=\emptyset$.

This updating of the candidate functions and sets is illustrated in Figure~\ref{fig:OPRfuncs} where the $Cost$ functions and $Set$ intervals are displayed  across two time steps. In this example a change in mean has been considered, using a least squares cost. The bold line on the left-hand graph corresponds to the function $Cost^*_t(\mu)$ and is made up of 7 pieces which relate to 6 candidate last changepoints. As the next time point is analysed the six $Cost_t^\tau(\mu)$ functions are updated using the formula $Cost_{t+1}^\tau(\mu)=Cost_t^\tau(\mu)+\gamma(y_{t+1},\mu)$ and a new function, $Cost_{t+1}^{t+1}(\mu)=F(t+1)+\beta$, is introduced corresponding to placing a changepoint at time $t+1$ (see middle plot). The functions which are no longer optimal for any values of $\mu$ (i.e. do not form any part of $Cost^*_{t+1}(\mu)$) can then be pruned, and one such function is removed in the right-hand plot. 
\begin{figure}[t]
  \centering
  \subfloat[End Time $t$]{
  \includegraphics[page=1,width=5.5cm,trim=1cm .5cm .5cm 1.5cm]{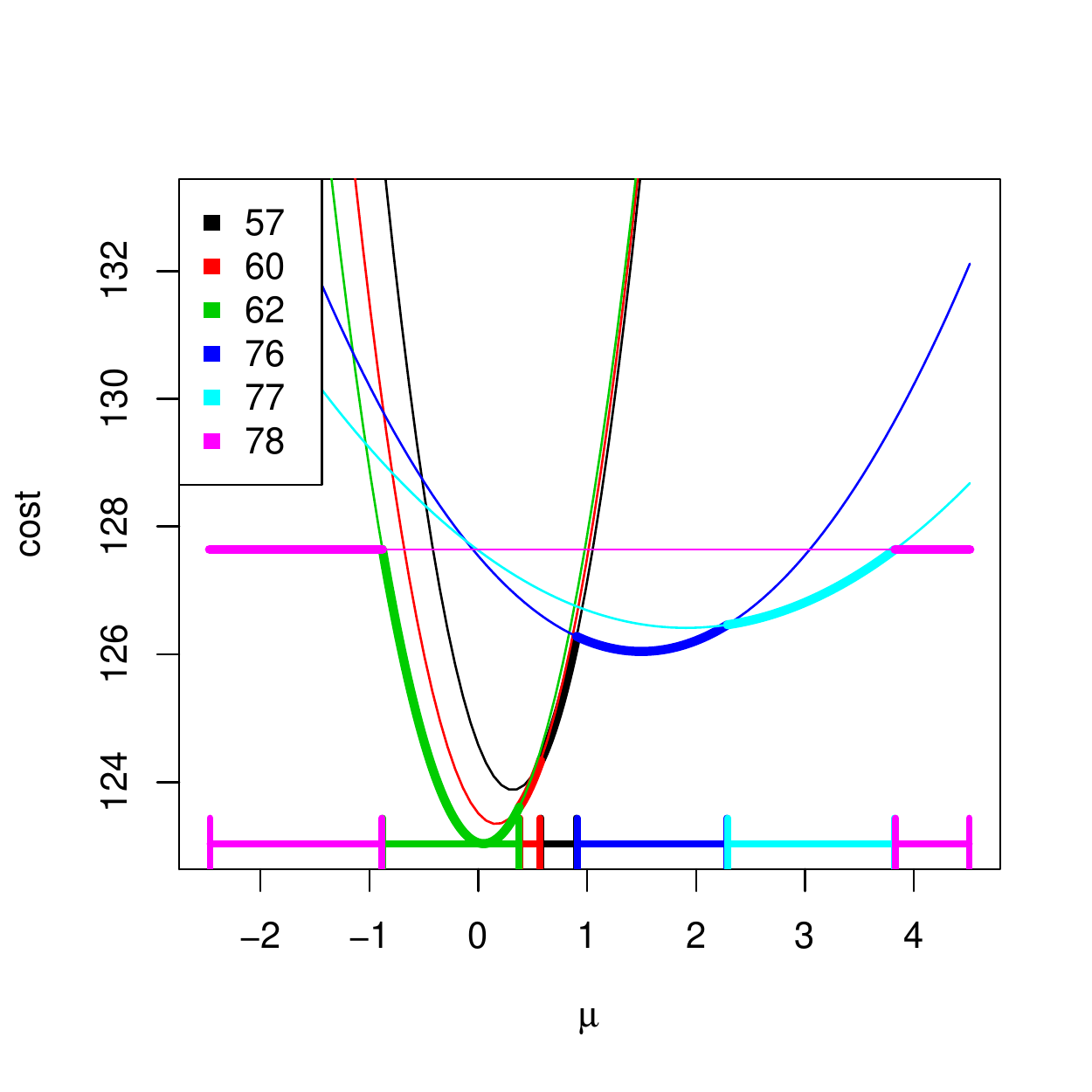}
} 
 \subfloat[Middle Time $t+1$]{
\includegraphics[page=2,width=5.5cm,trim=.75cm .5cm .75cm 1.5cm]{FPOP/FPOPcandidates.pdf}
}  
 \subfloat[End Time $t+1$]{
\includegraphics[page=3,width=5.5cm,trim=.5cm .5cm 1cm 1.5cm]{FPOP/FPOPcandidates.pdf}
}    
  \caption{Candidate functions over two time steps, the intervals shown along the bottom correspond to the intervals of $\mu$ for which each candidate last changepoint is optimal. 
  When $t=78$ (a) 6 candidates are optimal for some interval of $\mu$, however at $t=79$ (b), when the candidate functions are updated and the new candidate is added, 
  then candidate $\tau=78$ is no longer optimal for any $\mu$ and hence can be pruned (c).}
  \label{fig:OPRfuncs}
\end{figure} 

Once again we denote the set of potential last changes to consider as $R_{t}$ and then restrict the update rules (\ref{eq:4}) and (\ref{eq:6}) to $\tau\in R_{t}$. 
This set can then be recursively updated at each time step
\begin{align}
  R_{t+1} = \{ \tau\in\{R_{t} \cup \{t\}\} : Set_{t}^\tau \neq \emptyset \}.
\end{align}
These steps can then be applied directly to an Optimal Partitioning algorithm to form the FPOP method and the full pseudocode for this is presented in Algorithm~\ref{algo_OPR}.

\IncMargin{1em}
\begin{algorithm}[ht]
\SetKwData{Left}{left}\SetKwData{This}{this}\SetKwData{Up}{up}
\SetKwFunction{Union}{Union}\SetKwFunction{FindCompress}{FindCompress}
\SetKwInOut{Input}{Input}\SetKwInOut{Output}{Output}
\Input{Set of data of the form $\mathbf{y}_{1:n}=(y_1,\hdots,y_n)$,\\
A measure of fit $\gamma(\cdot,\cdot)$ dependent on the data and the mean,\\
A penalty $\beta$ which does not depend on the number or location of the changepoints.}
\BlankLine
Let $n=$length of data, and set $F(0)=-\beta$, $cp(0)=0$\;
then let $R_1=\{0\}$\;
and set $D=$ the range of $\mu$\;
$Set^0_{0}=D$\;
$Cost^0_0(\mu)=F(0)+\beta=0$\;
\For{$t=1,\hdots,n$}{
\For{$\tau \in R_t$}{
$Cost^\tau_{t}(\mu)=Cost^\tau_{t-1}(\mu)+\gamma(y_{t},\mu)$\;
}
Calculate $F(t)=\min_{\tau\in R_t}(\min_{\mu\in Set^{\tau}_t}[Cost^\tau_{t}(\mu)])$\;
Let $\tau_t=\operatorname*{arg\,min}_{\tau\in R_t}(\min_{\mu\in Set^{\tau}_t}[Cost^\tau_{t}(\mu)])$\;
Set $cp(t)=(cp(\tau_t),\tau_t)$\;
$Cost^t_{t}(\mu)=F(t)+\beta$\;
$Set^{t}_{t}=D$\;
\For{$\tau \in R_t$}{
$I^{\tau}_{t}=\{\mu: Cost^\tau_{t}(\mu)\leq F(t)+\beta\}$\;
$Set^\tau_{t}=Set^{\tau}_{t-1}\cap I^{\tau}_{t}$\;
$Set^t_{t}=Set^{t}_{t}\backslash I^{\tau}_{t}$\;
}
$R_{t+1} = \{ \tau\in\{R_{t} \cup \{t\}\} : Set_{t}^\tau \neq \emptyset \}$\;
}
\Output{The changepoints recorded in $cp(n)$.}
\BlankLine
\caption{Functional Pruning Optimal Partitioning (FPOP)}\label{algo_OPR}
\end{algorithm}\DecMargin{1em}

\subsection{Segment Neighbourhood with Inequality Pruning}
\label{sec:snp}

In a similar vein to Section~\ref{sec:opr}, Segment Neighbourhood Search can also benefit from using pruning methods. In Section~\ref{sec:segm-neighb-search} the method pDPA was discussed as a fast 
pruned version of Segment Neighbourhood Search. In this section a new method, Segment Neighbourhood with Inequality Pruning (SNIP), will be introduced. This takes the Segment Neighbourhood Search 
algorithm and uses inequality based pruning to increase the speed.

Under condition (C2) the following result can be proved for Segment Neighbourhood Search and this will enable points to be pruned from the candidate changepoint set. 
\begin{theorem}
  \label{SNPtheorem}
  Assume that there exists a constant, $\kappa$, such that condition C2 holds. If, for any $k\geq 1$ and $t<s$
\begin{align}
  C_{k-1,t}+\mathcal{C}(\mathbf{y}_{t+1:s})+\kappa> C_{k-1,s}\label{SNPtheorem2}
\end{align}
then at any future time $T>s$, $t$ cannot be the position of the last changepoint in the optimal segmenation of $y_{1:T}$ with $k$ changepoints.
\end{theorem}
\begin{proof} \label{sec:assume-that-there} 
The idea of the proof is to show that a segmentation of $y_{1:T}$ into $k$ segments with the last changepoint at $s$ will be better than one 
with the last changepoint at $t$ for all $T>s$.
 
 Assume that (\ref{SNPtheorem2}) is true. Now for any $s<T\leq n$
  \begin{align*}
C_{k-1,t}+\mathcal{C}(\mathbf{y}_{t+1:s})+\kappa+&>C_{k-1,s},\\
    C_{k-1,t}+\mathcal{C}(\mathbf{y}_{t+1:s})+\kappa+\mathcal{C}(\mathbf{y}_{s+1:T})&>C_{k-1,s}+\mathcal{C}(\mathbf{y}_{s+1:T}),\\
C_{k-1,t}+\mathcal{C}(\mathbf{y}_{t+1,T})&>C_{k-1,s}+\mathcal{C}(\mathbf{y}_{s+1,T}),\hspace{50pt}\mbox{     (by C2).}
  \end{align*}
Therefore for any $T>s$ the cost $C_{k-1,t}+\mathcal{C}(\mathbf{y}_{t+1,T})>C_{k,T}$ and hence $t$ cannot be the optimal location of the last changepoint when segmenting $\mathbf{y}_{1:T}$ with $k$ changepoints. 
\end{proof}

Theorem~\ref{SNPtheorem} implies that the update rule (\ref{eq:SNS}) can be restricted to a reduced set over $\tau$ of potential last changes to consider. This set, which we shall denote as $R_{k,t}$, can be updated simply by
\begin{align}
R_{k,t+1}=\{v\in \{R_{k,t}\cup\{t\}\}:C_{k-1,v}+\mathcal{C}(\mathbf{y}_{v+1,t})+\kappa< C_{k-1,t}\}.
\end{align}

This new algorithm, SNIP, is described fully in Algorithm~\ref{algo_segneighpelt}.

\IncMargin{1em}
\begin{algorithm}[ht]
\SetKwData{Left}{left}\SetKwData{This}{this}\SetKwData{Up}{up}
\SetKwFunction{Union}{Union}\SetKwFunction{FindCompress}{FindCompress}
\SetKwInOut{Input}{Input}\SetKwInOut{Output}{Output}
\Input{Set of data of the form $\mathbf{y}_{1:n}=(y_1,\hdots,y_n)$,\\
A measure of fit $\mathcal{C}(\cdot)$ dependent on the data (needs to be minimised),\\
An integer, $K$, specifying the maximum number of changepoints to find,\\
A constant $\kappa$ that satisfies: $\mathcal{C}(\mathbf{y}_{t+1:s})+\mathcal{C}(\mathbf{y}_{s+1:T})+\kappa\leq \mathcal{C}(\mathbf{y}_{t+1:T})$.}
\BlankLine
Let $n=$length of data\;
Set $C_{0,t}=\mathcal{C}(\mathbf{y}_{1:t})$, for all $t\in \{1,\hdots,n\}$\;
\For{$k=1,\hdots,K$}{
Set $R_{k,k+1}=\{k\}$.
\For{$t=k+1,\hdots,n$}{
Calculate $C_{k,t}=\min_{v\in R_{k,t}}(C_{k-1,v}+\mathcal{C}(\mathbf{y}_{v+1:t}))$\;
Set $R_{k,t+1}=\{v\in \{R_{k,t}\cup\{t\}\}:C_{k-1,v}+\mathcal{C}(\mathbf{y}_{v+1,t})+\kappa< C_{k-1,t}\}$\;
}
Set $\tau_{k,1}=\operatorname*{arg\,min}_{v\in R_{k,n}}(C_{k-1,v}+\mathcal{C}(\mathbf{y}_{v+1,n}))$\;
\For{$i=2,\hdots,k$}{
Let $\tau_{k,i}=\operatorname*{arg\,min}_{v\in R_{k-i,\tau_{k,i-1}}}(C_{k-1,v}+\mathcal{C}(\mathbf{y}_{v+1,\tau_{k,i-1}}))$\;
}
}
\Output{For $k=0,\hdots,K$: the total measure of fit, $C_{k,n}$, for $k$ changepoints and the location of the changepoints for that fit, $\boldsymbol{\tau}_{k,(1:k)}$.}
\BlankLine
\caption{Segment Neighbourhood with Inequality Pruning (SNIP)}\label{algo_segneighpelt}
\end{algorithm}\DecMargin{1em}

\section{Comparisons Between Pruning Methods}
\label{sec:simil-diff-betw}


Functional and inequality based pruning both offer increases in the efficiency in solving both the penalised and constrained problems, 
however their use depends on the assumptions which can be made on the cost function. 
Inequality based pruning is dependent on the assumption C2, while functional pruning requires the slightly stronger condition C1. 
Furthermore, inequality based pruning has a large computational overhead, and is currently only feasible for detecting a change in a univariate parameter.

If we consider models for which both pruning methods can be implemented, we can compare the extent to which the methods prune. This will give some insight into when
the different pruning methods would be expected to work well.
To explore this in Figures \ref{fig:noleftPELT-OPR} and \ref{fig:SNPPDPAcandidates} we look at the amount of candidates stored by functional and inequality based pruning in each of the two optimisation problems.

\begin{figure}[t]
  \centering
 \includegraphics[width=8cm, trim=1cm 1cm 1cm 1cm]{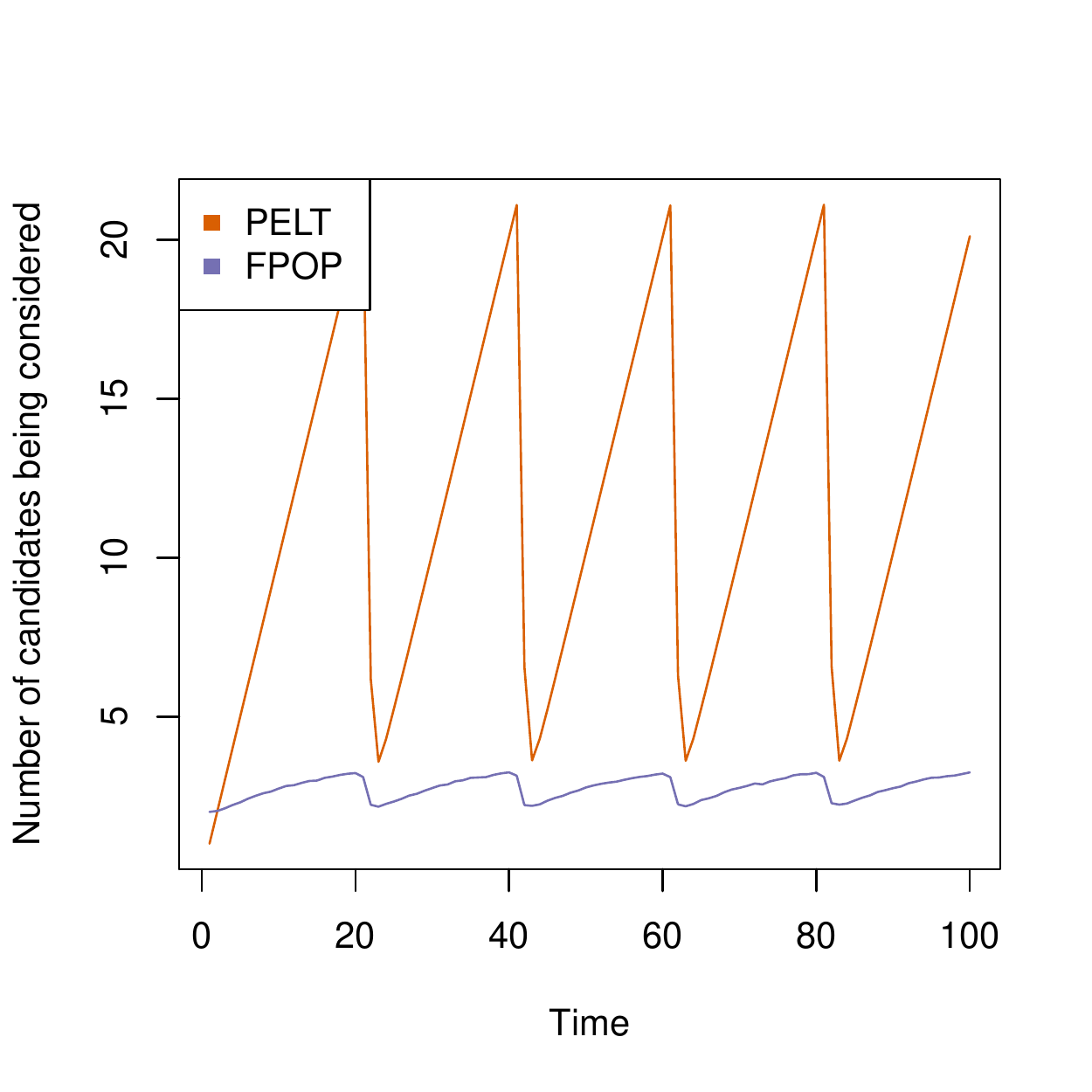}
\caption{Comparison of the number of candidate changepoints stored over time by FPOP (purple) and PELT (orange). Averaged over 1000 data sets with changepoints at $t=20,40,60$ and $80$.}
  \label{fig:noleftPELT-OPR}
\end{figure}

As Figure~\ref{fig:noleftPELT-OPR} illustrates, PELT prunes very rarely; only when evidence of a change is particularly high. 
In contrast, FPOP prunes more frequently keeping the candidate set small throughout.
Figure~\ref{fig:SNPPDPAcandidates} shows similar results for the constrained problem. While pDPA constantly prunes, 
SNIP only prunes sporadically. In addition SNIP fails to prune much at all for low values of $k$. 

\begin{figure}[t]
  \centering
   \subfloat{
     \includegraphics[page=1,width=4.5cm, trim=0cm 1cm 0cm 1cm]{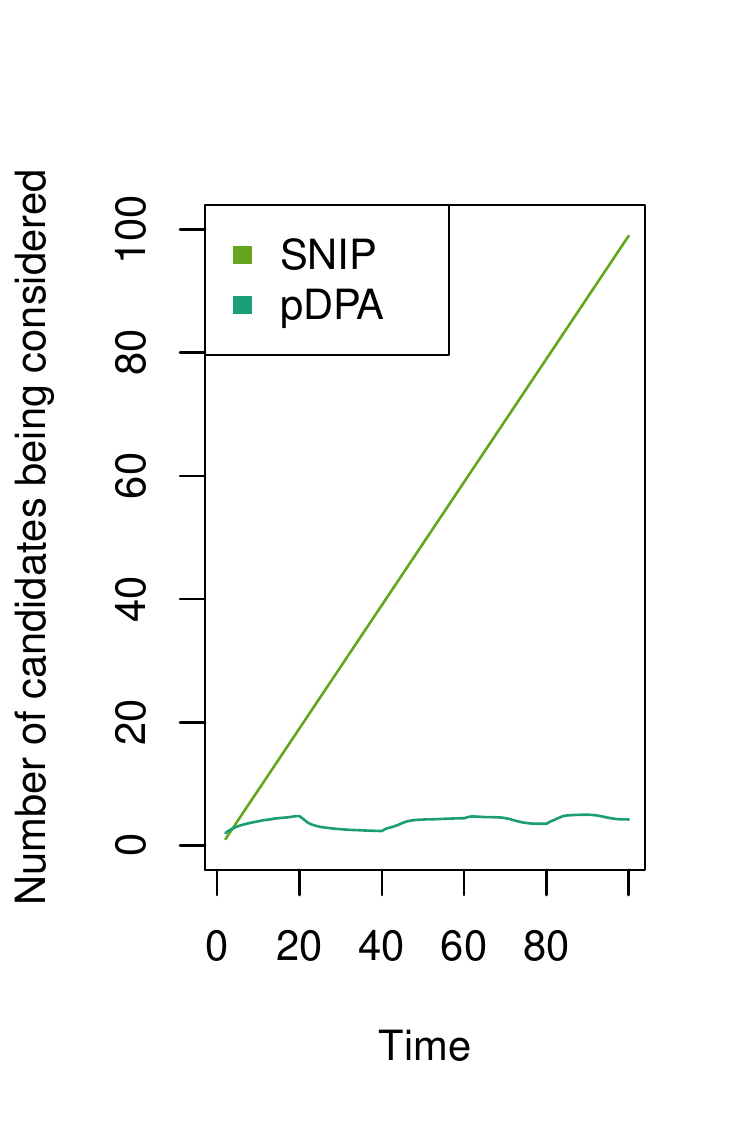} 
   }
 \subfloat{
     \includegraphics[page=2,width=4.5cm, trim=0cm 1cm 0cm 1cm]{NumberOfCPs/ConsideredSNIPandPDPA.pdf} 
   }
 \subfloat{
     \includegraphics[page=3,width=4.5cm, trim=0cm 1cm 0cm 1cm]{NumberOfCPs/ConsideredSNIPandPDPA.pdf} 
   }
 \subfloat{
     \includegraphics[page=4,width=4.5cm, trim=0cm 1cm 0cm 1cm]{NumberOfCPs/ConsideredSNIPandPDPA.pdf} 
   }
  \caption{Comparison of the number of candidate changepoints stored over time by pDPA (teal) and SNIP (green) at multiple values of $k$ in the algorithms (going from  left to right $k=2,3,4,5$). Averaged over 1000 data sets with changepoints at $t=20,40,60$ and $80$.}
  \label{fig:SNPPDPAcandidates}
\end{figure}

Figures \ref{fig:noleftPELT-OPR} and \ref{fig:SNPPDPAcandidates} give strong empirical evidence that functional pruning prunes more points than the inequality based method. 
In fact it can be shown that any point pruned by inequality based pruning will also be pruned at the same time step by functional pruning. 
This result holds for both the penalised and constrained case and is stated formally in Theorem \ref{thr:1}.

\begin{theorem}
  \label{thr:1}
  Let $\mathcal{C}(\cdot)$ be a cost function that satisfies condition C1, and consider solving either the constrained or penalised
  optimisation problem using dynamic programming and either inequality or functional pruning. 
  
  Any point pruned by inequality based pruning at time $t$ will also have been pruned by functional pruning at the same time.
\end{theorem}
\begin{proof}
  \label{sec:given-cost-function-1}
We prove this for pruning of optimal partitioning, with the ideas extending directly to the pruning of the Segment Neighbourhood algorithm.
  
  For a cost function which can be decomposed into pointwise costs, it's clear that condition C2 holds when $\kappa=0$ and hence inequality based pruning can be used.
 Recall that the point $\tau$ (where $\tau<t$, the current time point) is pruned by inequality based pruning in the penalised case if
  \begin{align*}
    F(\tau)+\mathcal{C}(\mathbf{y}_{\tau+1:t})\geq F(t),
  \end{align*}
Then, by letting $\hat{\mu}_{\tau}$ be the value of $\mu$ such that $Cost^\tau_t(\mu)$ is minimised, this is equivalent to
\begin{align*}
  Cost^\tau_t(\hat{\mu}_{\tau})-\beta\geq F(t),
\end{align*}
Which can be generalised for all $\mu$ to
\begin{align*}
  Cost^\tau_t(\mu)\geq F(t)+\beta.
\end{align*}
Therefore equation~\eqref{eq:5} holds for no value of $\mu$ and hence $I^\tau_t=\emptyset$ and furthermore $Set^\tau_{t+1}=Set^\tau_t\cap I^\tau_t=\emptyset$ meaning that $\tau$ is pruned under functional pruning.
\end{proof}

\section{Empirical evaluation of FPOP}
\label{sec:emp-eval-fpop}

As explained in Section \ref{sec:simil-diff-betw} functional pruning
leads to a better pruning in the following sense: any point pruned by
inequality based pruning will also be pruned by functional pruning.
However, functional pruning is computationally more demanding than
inequality based pruning.  We thus decided to empirically compare the
performance of FPOP to PELT, pDPA and Binary Segmentation (Binseg).

To do so, we implement FPOP for the quadratic loss in C++.
More precisely, we consider the quadratic cost (\ref{eq:Cost}). 
We assess the runtimes of FPOP on both real microarray data as well as synthetic data.
All algorithms were implemented in C++. For pDPA and Binary Segmentation we had to input
a maximum number of changepoints to search for, which we denote $K$ as before. 


Note that the Binary Segmentation heuristic is computationally
extremely fast. That is because its complexity is on average $\mathcal{O}(n \log K)$. Furthermore, it relies on a few
fairly simple operations and hence the constant in front of this big
$\mathcal O$ notation is expected to be quite low.  For this reason we
believe that Binary Segmentation is good reference point in terms of
speed and we do not think it is possible to be much faster. However, unlike the other algorithms,
Binary Segmentation is not guaranteed to find the optimal segmentation. For a fuller investigation
of the loss of accuracy that can occur when using Binary Segmentation see \cite{Killick2012a}.

\subsection{Speed benchmark: 4467 chromosomes from 
  tumour microarrays}

\citet{Hocking2014} proposed to benchmark the speed of 
segmentation algorithms on a database of 4467 problems of size varying
from 25 to 153662 data points. These data come from different
microarrays data sets (Affymetrix, Nimblegen, BAC/PAC) and different
tumour types (leukaemia, lymphoma, neuroblastoma, medulloblastoma).

We compared FPOP to several other segmentation algorithms: pDPA \citep{Rigaill2010}, PELT \citep{Killick2012a}, and Binary Segmentation
(Binseg).  For such data, we expect a small number of changes and for
all profiles we ran pDPA and Binseg with a maximum number of changes
$K=52$.  

We used the R \verb|system.time| function to measure
the execution time of all 4 algorithms on all 4467 segmentation
problems. The R source code for these timings is in
\verb|benchmark/systemtime.arrays.R| in the opfp project repository on
R-Forge:
\url{https://r-forge.r-project.org/projects/opfp/}

It can be seen in Figure~\ref{fig:sys_runtimes_microarray} that in
general FPOP is faster than PELT and pDPA, but about two times slower
than Binseg.  Note that \verb|system.time| is inaccurate for small
times due to rounding, as pointed out in Figure~\ref{fig:sys_runtimes_microarray}
middle and right.  For these small profiles we also assessed the
performances of FPOP, PELT and Binseg using the \verb|microbenchmark|
package and confirmed that FPOP is faster than PELT. For these small
problems we were surprised to observe that FPOP exhibited about the
same speed as Binseg (microbenchmark results figure not shown).
 
\begin{figure}[t]
 \parbox{6.5cm}{ \includegraphics[height=5cm]{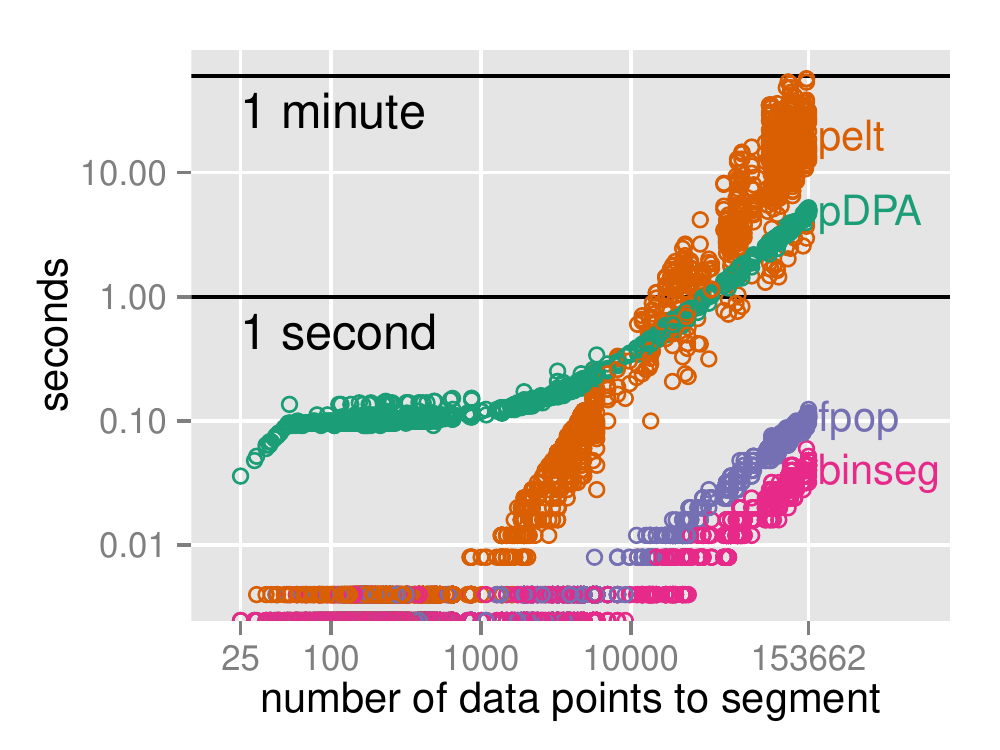}
  } \parbox{12.5cm}{  \includegraphics[width=\linewidth]{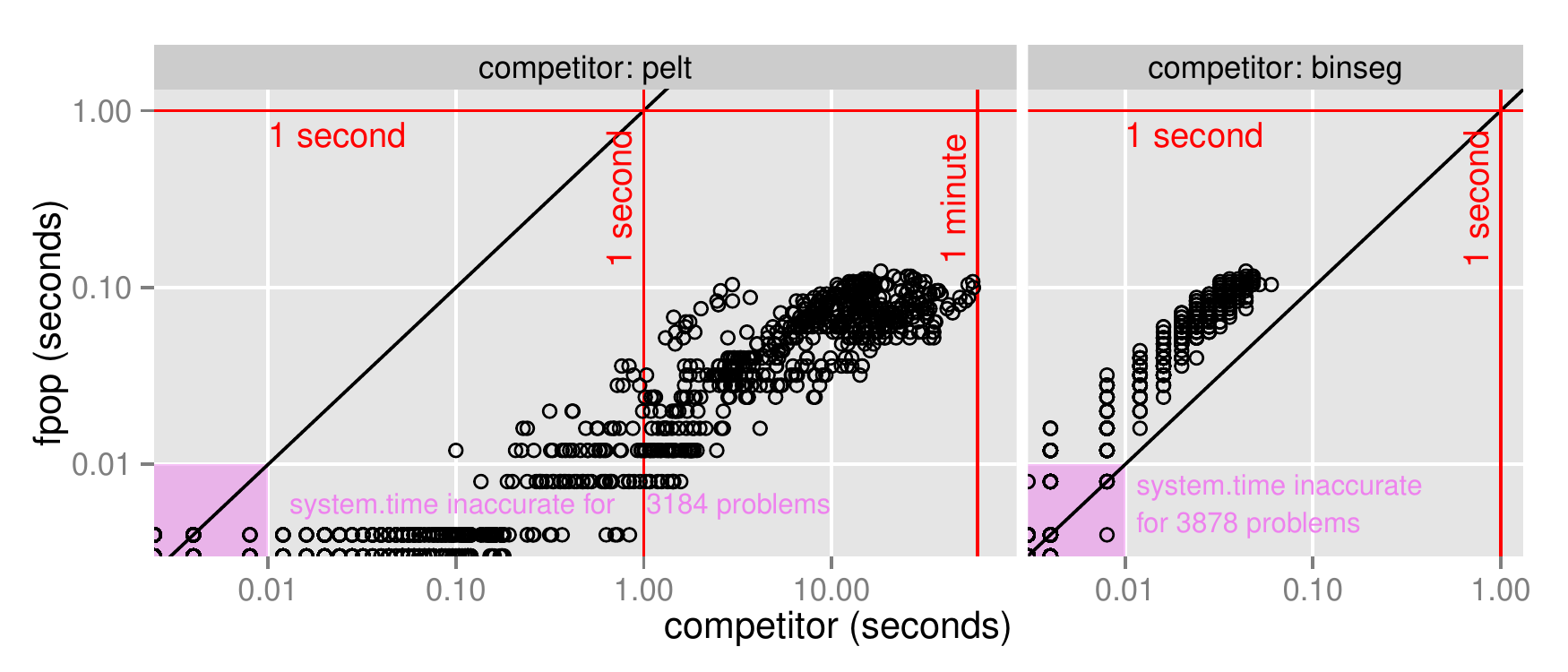}}

\caption{(Left) Runtimes of FPOP, PELT,
pDPA and Binseg as a function of the length of the profile  using  on the tumour micro array benchmark. 
(Middle) Runtimes of PELT and FPOP for the same profiles.
(Right) Runtimes of Binseg and FPOP for the same profiles.
}\label{fig:sys_runtimes_microarray}
\end{figure}

\subsection{Speed benchmark: simulated data with different 
  number of changes}

The speed of PELT, Binseg and pDPA depends on the underlying number of changes.
For pDPA and Binseg the relationship is clear as to cope with a larger number of changes
one need to increase the maximum number of changes to look for $K$. 
For a fixed size signal the runtime dependency is expected to be in $\mathcal{O}(\log K)$ for Binseg and
in $\mathcal{O}(K)$ for pDPA.

For PELT the relationship is less clear, however we expect pruning to be more efficient
if there is a large number of changepoints. Hence for a fixed size signal we expect 
the runtime of PELT to improve with the underlying number of changes.

Based on Section \ref{sec:simil-diff-betw} we expect FPOP to be more efficient than PELT and pDPA.
Thus it seems reasonable to expect FPOP to be efficient for the whole range of $K$. 
This is what we empirically check in this section.

To do that we simulated Gaussian signal with 200000 data points and varied the number of changes
and timed the algorithms for each of these signals. We then repeat the same experience for signals with $10^7$ and timed FPOP and Binseg only.
The R source code
for these timings is in \verb|benchmark/systemtime.simulation.R| in
the opfp project repository on R-Forge: \url{https://r-forge.r-project.org/projects/opfp/}.

It can be seen in Figure \ref{fig:simu_numberK} that FPOP is always faster than pDPA and PELT. 
Interestingly for both $n=2\times 10^5$ and $10^7$ it is faster than Binseg for a true number of changepoints larger than 500.

\begin{figure}[t]
  \begin{center}
      \includegraphics[width=8cm]{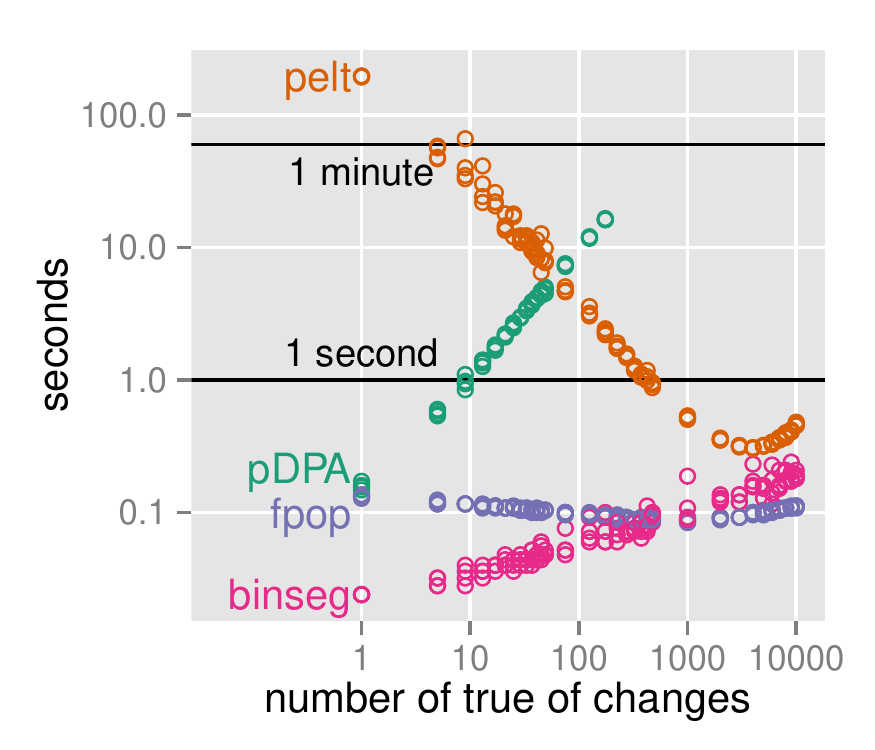}
    \includegraphics[width=8cm]{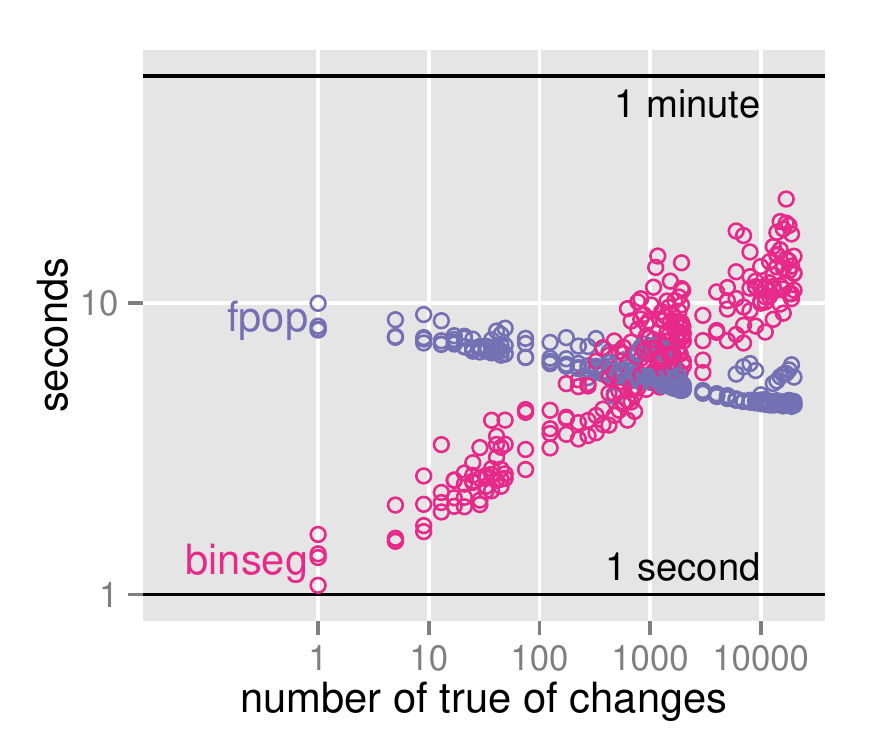}
  \end{center}
\caption{Runtimes as a function of the true number of changepoints.
(Left) For FPOP, PELT,
pDPA and Binseg and $n=2 \times 10^5$
(Right) For Binseg and FPOP and $n=10^7$
}\label{fig:simu_numberK}
\end{figure}

\subsection{Accuracy benchmark: the neuroblastoma data set}

\citet{Hocking2013} proposed to benchmark the changepoint
detection accuracy of segmentation models by using annotated regions
defined by experts when they visually inspected scatterplots of the
data. The neuroblastoma data are a set of 575 copy number microarrays
of neuroblastoma tumours, and each chromosome is a separate
segmentation problem. The benchmark is to use a training set of $n$
annotated chromosomes to learn a segmentation model, and then quantify
the annotation error on a test set. Let $d_1, \dots, d_n$ be the
number of data points to segment on each chromosome in the training
data set, and let $\mathbf y_1\in\RR^{d_1}, \dots, \mathbf
y_n\in\RR^{d_n}$ be the vectors of noisy data for each chromosome in
the training set.

Both PELT and pDPA have been applied to this benchmark by first
defining $\beta = \lambda d_i$ for all $i\in\{1, \dots, n\}$, and then
choosing the constant $\lambda$ that maximises agreement with the
annotated regions in the training set. 

Since FPOP computes the same
segmentation as PELT, we obtain the same error rate for
FPOP on the neuroblastoma benchmark.
As shown on
\url{http://cbio.ensmp.fr/~thocking/neuroblastoma/accuracy.html}, FPOP
(fpop) achieves 2.2\% test error. This is the smallest error across all
algorithms that have been tested.





\section{Discussion}

We have introduced two new algorithms for detecting changepoints, FPOP and SNIP. A natural question is which of these, and the existing algorithms, pDPA and PELT, should be used in which applications. 
There are two stages to answering this question. The first is whether to detect changepoints through solving the constrained or the penalised optimisation problem, and the second is whether to
use functional or inequality based pruning.

The advantage of solving the constrained optimisation problem is that this gives optimal segmentations for a range of numbers of changepoints. The disadvantage is that solving it is slower than
solving the penalised optimisation problem, particularly if there are many changepoints. In interactive situations where you wish to explore segmentations of the data, then solving the 
constrained problem is to be preferred \cite[]{Hocking2014}. However in non-interactive scenarios when the penalty parameter is known in advance, it will be faster to solve the penalised problem
 to recover the single segmentation of interest.

The decision as to which pruning method to use is purely one of computational efficiency. We have shown that functional pruning always prunes more than inequality based pruning, and empirically have seen 
that this difference can be large, particularly if there are few changepoints. However functional pruning can be applied less widely. Not only does it require a stronger condition on the cost functions, but currently its implementation
has been restricted to detecting changes in a uni-variate parameter from a model in the exponential family. Even for situations where functional pruning can be applied, its computational overhead per non-pruned candidate is higher.

Our experience suggests that you should prefer functional pruning in the situations where it can be applied. For example FPOP was always faster than PELT for detecting a change in mean in the empirical studies we conducted,
the difference in speed is particularly large in situations where there are few changepoints. Furthermore we observed FPOP's computational speed was robust to changes in the number of changepoints to be detected, and was
even competitive with, and sometimes faster than, Binary Segmentation. 


{\bf Acknowledgements} We thank Adam Letchford for helpful comments and discussions. This research was supported by EPSRC grant EP/K014463/1. Maidstone gratefully acknowledges funding from EPSRC via the
STOR-i Centre for Doctoral Training.

\bibliography{refs}

\begin{thebibliography}{}

\bibitem[Akaike, 1974]{Akaike1974}
Akaike, H. (1974).
\newblock {A New Look at the Statistical Model Identification}.
\newblock {\em IEEE Transactions on Automatic Control}, 19:716--723.

\bibitem[Auger and Lawrence, 1989]{Auger1989}
Auger, I.~E. and Lawrence, C.~E. (1989).
\newblock {Algorithms for the Optimal Identification of Segment Neighborhoods}.
\newblock {\em Bulletin of Mathematical Biology}, 51:39--54.

\bibitem[Braun et~al., 2000]{Braun/Braun/Muller:2000}
Braun, J.~V., Braun, R.~K., and Muller, H.~G. (2000).
\newblock {Multiple Changepoint Fitting via Quasilikelihood, With Application
  to {DNA} Sequence Segmentation}.
\newblock {\em Biometrika}, 87:301--314.

\bibitem[Cleynen et~al., 2012]{Cleynen2012}
Cleynen, A., Koskas, M., and Rigaill, G. (2012).
\newblock {A Generic Implementation of the Pruned Dynamic Programing
  Algorithm}.
\newblock {\em ArXiv e-prints}.

\bibitem[Davis et~al., 2006]{Davis2006}
Davis, R.~A., Lee, T. C.~M., and Rodriguez-Yam, G.~A. (2006).
\newblock {Structural Break Estimation for Nonstationary Time Series Models}.
\newblock {\em Journal of the American Statistical Association}, 101:223--239.

\bibitem[Fryzlewicz, 2012]{Fryzlewicz2012}
Fryzlewicz, P. (2012).
\newblock {Wild Binary Segmentation for Multiple Change-Point Detection}.
\newblock {\em Annals of Statistics, to appear}.

\bibitem[Hocking et~al., 2014]{Hocking2014}
Hocking, T.~D., Boeva, V., Rigaill, G., Schleiermacher, G., Janoueix-Lerosey,
  I., Delattre, O., Richer, W., Bourdeaut, F., Suguro, M., Seto, M., Bach, F.,
  and Vert, J.-P. (2014).
\newblock {SegAnnDB: Interactive Web-Based Genomic Segmentation}.
\newblock {\em Bioinformatics}, 30:1539--46.

\bibitem[Hocking et~al., 2013]{Hocking2013}
Hocking, T.~D., Schleiermacher, G., Janoueix-lerosey, I., Boeva, V., Cappo, J.,
  Delattre, O., Bach, F., and Vert, J.-P. (2013).
\newblock {Learning Smoothing Models of Copy Number Profiles Using Breakpoint
  Annotations}.
\newblock {\em BNC Bioinformatics}, 14.

\bibitem[Jackson et~al., 2005]{Jackson2005}
Jackson, B., Scargle, J.~D., Barnes, D., Arabhi, S., Alt, A., Gioumousis, P.,
  Gwin, E., Sangtrakulcharoen, P., Tan, L., and Tsai, T.~T. (2005).
\newblock {An Algorithm for Optimal Partitioning of Data on an Interval}.
\newblock {\em {IEE Signal Processing Letters}}, 12:105--108.

\bibitem[Killick et~al., 2012]{Killick2012a}
Killick, R., Fearnhead, P., and Eckley, I.~A. (2012).
\newblock {Optimal Detection of Changepoints With a Linear Computational Cost}.
\newblock {\em Journal of the American Statistical Association},
  107:1590--1598.

\bibitem[Lavielle, 2005]{Lavielle2005}
Lavielle, M. (2005).
\newblock {Using Penalized Contrasts for the Change-Point Problem}.
\newblock {\em Signal Processing}, 85:1501--1510.

\bibitem[Olshen et~al., 2004]{Olshen2004}
Olshen, A.~B., Venkatraman, E.~S., Lucito, R., and Wigler, M. (2004).
\newblock {Circular Binary Segmentation for the Analysis of Array-Based DNA
  Copy Number Data}.
\newblock {\em Biostatistics}, 5:557--572.

\bibitem[Picard et~al., 2011]{Picard2011}
Picard, F., Lebarbier, E., Hoebeke, M., Rigaill, G., Thiam, B., and Robin, S.
  (2011).
\newblock {Joint Segmentation, Calling, and Normalization of Multiple CGH
  Profiles.}
\newblock {\em Biostatistics}, 12:413--428.

\bibitem[Reeves et~al., 2007]{Reeves2007}
Reeves, J., Chen, J., Wang, X.~L., Lund, R., and Lu, Q.~Q. (2007).
\newblock {A Review and Comparison of Changepoint Detection Techniques for
  Climate Data}.
\newblock {\em Journal of Applied Meteorology and Climatology}, 46:900--915.

\bibitem[{Rigaill}, 2010]{Rigaill2010}
{Rigaill}, G. (2010).
\newblock {Pruned Dynamic Programming for Optimal Multiple Change-Point
  Detection}.
\newblock {\em ArXiv e-prints}.

\bibitem[Schwarz, 1978]{Schwarz1978}
Schwarz, G. (1978).
\newblock {Estimating the Dimension of a Model}.
\newblock {\em {The Annals of Statistics}}, 6:461--464.

\bibitem[Scott and Knott, 1974]{Scott1974}
Scott, A.~J. and Knott, M. (1974).
\newblock {A Cluster Analysis Method for Grouping Means in the Analysis of
  Variance}.
\newblock {\em Biometrics}, 30:507--512.

\bibitem[Zhang and Siegmund, 2007]{Zhang2007}
Zhang, N.~R. and Siegmund, D.~O. (2007).
\newblock {A Modified Bayes Information Criterion With Applications to the
  Analysis of Comparative Genomic Hybridization Data}.
\newblock {\em Biometrics}, 63:22--32.

\end{thebibliography}

\end{document}